\documentclass{amsart}
\usepackage[a4paper,margin=25mm]{geometry}
\usepackage{amssymb}
\usepackage{amsthm}

\usepackage{tikz,hyperref}
\usepackage{bbold}

\def\<#1>{\langle#1\rangle}

\def\eatspace#1{#1}
\def\step#1#2{%
  \par\kern1pt\dimen144=#2em\advance\dimen144by1.67em
  \hangindent=\dimen144\hangafter=1
  \leavevmode\rlap{\small#1}\kern\dimen144\relax\eatspace}

\newcommand{\N}{\mathbb{N}}

\newcommand{\Z}{\mathbb{Z}}
\newcommand{\rank}{\operatorname{rank}}

\renewcommand{\Gamma}{\varGamma}
\newcommand{\A}{\mathbf{A}}
\newcommand{\B}{\mathbf{B}}
\def\C{\mathbf{C}}

\newcommand{\gl}{\operatorname{GL}}

\newcommand{\M}{\mathcal{M}}
\renewcommand{\l}{\ell}

\newtheorem{theorem}{Theorem}

\newtheorem{lemma}[theorem]{Lemma}
\newtheorem{definition}[theorem]{Definition}

\newtheorem{algorithm}{Algorithm}

\begin{document}

 \title{Flip Graphs with Symmetry and New Matrix Multiplication Schemes}

 \author[J. Moosbauer]{Jakob Moosbauer$^{\ast}$}
 \address{Jakob Moosbauer, Department of Computer Science, University of Warwick,
   Coventry, United Kingdom}
 \email{jakob.moosbauer@warwick.ac.uk}
 \author[M. J. Poole]{Michael Poole}
 \address{Michael J. Poole, Independent Researcher, United Kingdom}
  \email{mikepoole73@googlemail.com}

  \thanks{$\ast$ Jakob Moosbauer was supported by the EPSRC grant EP/W014882/2}

 \begin{abstract}
   The flip graph algorithm is a method for discovering new matrix
   multiplication schemes by following random walks on a graph.  We introduce a
   version of the flip graph algorithm for matrix multiplication schemes that
   admit certain symmetries.  This significantly reduces the size of the search
   space, allowing for more efficient exploration of the flip graph.  The
   symmetry in the resulting schemes also facilitates the process of lifting
   solutions from $F_{2}$ to $\mathbb{Z}$.  Our results are new schemes for
   multiplying $5\times 5$ matrices using $93$ multiplications and $6\times 6$
   matrices using $153$ multiplications over arbitrary ground fields.
 \end{abstract}

 \maketitle

 \section{Introduction}\label{sec:introduction}

 Matrix multiplication lies at the heart of many problems in scientific
 computing and theoretical computer science. Since Strassen’s seminal work~\cite{st:gein}
 demonstrating that $2 \times 2$ matrices can be multiplied using only seven
 scalar multiplications, researchers have sought more efficient algorithms both
 for small matrices and in the asymptotic regime.  Early improvements for small
 matrices include the work by Hopcroft and Kerr \cite{hk:omtn} for certain
 rectangular formats and Laderman’s algorithm \cite{la:anaf} for $3 \times 3$
 matrices. More recently, computer-assisted methods like numerical
 optimization~\cite{sm:tbca}, SAT solvers~\cite{hks:lsff,hks:nwtm}, reinforcement
 learning~\cite{fbh+:dfmm} and stochastic search methods~\cite{km:fgfm} have
 further expanded the catalog of known small matrix multiplication schemes.

 In parallel, a separate line of research has produced asymptotically fast
 algorithms that focus on minimizing the exponent $\omega$ of matrix
 multiplication.  Building on the groundbreaking work of Coppersmith and
 Winograd~\cite{cw:mmva} there has been a series of improvements leading to the
 current record of $\omega<2.371339$~\cite{adw+:mayf}.  While these algorithms
 provide the best asymptotic complexity, they only achieve lower exponents for
 matrices of astronomically large sizes.  

 Over the past decade, some authors worked on analyzing symmetries in matrix
 multiplication algorithms.  Ballard et al.~\cite{bilr:tgo2} combine numerical
 methods with symmetry constraints to find $3\times 3$ matrix multiplication
 schemes that have certain symmetries.  Grochow and Moore~\cite{gm:mmaf}
 describe an algorithm for multiplying $n\times n$ matrices in $n^{3}-n+1$
 multiplications based on symmetries of the matrix multiplication tensor.
 In~\cite{bu:somm} Burichenko analyses the symmetry groups of several known
 matrix multiplication algorithms and in~\cite{bu:oago} he proves a $23$
 multiplications lower bound for $3\times 3$ multiplication schemes that admit
 certain symmetries.

 A promising approach for exploring the space of small matrix multiplication
 schemes is provided by the flip graph algorithm \cite{km:fgfm}. In this
 framework, vertices correspond to distinct multiplication schemes and edges
 represent elementary transformations (``flips'') that modify a small number of
 rank-one tensors while preserving the overall rank. Certain moves, called
 ``reductions'', lower the rank and thus offer a path to improved
 algorithms. Random walks on the flip graph have already yielded new schemes in
 recent work~\cite{km:fgfm,km:snnm,aih:afga}.

 In this paper, we extend the flip graph approach by incorporating symmetry
 constraints into the search.  Building on the observation that many known
 schemes admit certain symmetry groups, we restrict our search to symmetric
 schemes.  This restriction dramatically reduces the size of the search space
 and facilitates the subsequent process of lifting solutions from the two
 element field to arbitrary ground fields.  In particular, our symmetry-based
 algorithm leads to new multiplication schemes for $5 \times 5$ matrices using
 93 multiplications and for $6 \times 6$ matrices using 153 multiplications,
 which improve upon previously known bounds $97$~\cite{km:fgfm} (or $94$ over
 $F_{2}$)~\cite{aih:afga} and $160$~\cite{sm:tbca}.

 Small matrix multiplication schemes also serve as building blocks for larger
 matrix multiplication schemes, as demonstrated by work of Drevet et
 al.\cite{dis:otfs} and Sedoglavic~\cite{se:anca,se:yaco}.  The previously best
 bound for $6\times 6$ matrices was achieved by combining four schemes for
 multiplying $3\times 3$ matrices with $3 \times 6$ matrices in $40$
 multiplications.  In contrast, our scheme for $6\times 6$ matrices is
 ``atomic'', in the sense that it does not originate from combining smaller
 schemes.

 \section{Matrix Multiplication and its Symmetry Group}
 \label{sec:matr-mult}

 For now, let $K$ be an arbitrary field and let $R$ be a $K$-algebra.  Let
 $\A,\B\in R^{n\times n}$.  The asymptotic complexity of computing the matrix
 product $\C=\A\B$ is, up to subpolynomial factors, determined by the number of
 multiplications needed to compute the entries of $\C$ without using
 commutativity.  This number is called the rank of $n\times n$ matrix
 multiplication.

 For example, Strassen's algorithm demonstrates that the rank of $2 \times 2$
 matrix multiplication is at most $7$.  Suppose
 \begin{equation*}
   \A = \begin{pmatrix}
         a_{1,1}&a_{1,2}\\
         a_{2,1}&a_{2,2}\\
       \end{pmatrix},
   \quad
   \B = \begin{pmatrix}
         b_{1,1}&b_{1,2}\\
         b_{2,1}&b_{2,2}
       \end{pmatrix}
   \quad\text{and}\quad
   \C = \begin{pmatrix}
         c_{1,1}&c_{1,2}\\
         c_{2,1}&c_{2,2}
       \end{pmatrix}.
     \end{equation*}

     Then Strassen's algorithm computes:
\begin{alignat*}{4}
   m_1 &= (a_{1,1} + a_{2,2}) (b_{1,1} + b_{2,2}) &\quad\smash{\raisebox{-6.4\baselineskip}{\rule\fboxrule{7\baselineskip}}}\quad c_{1,1} &= m_1 + m_4 + m_5 - m_6 \\
   m_2 &= (a_{1,1}) (b_{1,2} - b_{2,2}) & c_{1,2} &= m_2 + m_6 \\
   m_3 &= (a_{2,1} + a_{2,2}) (b_{1,1}) & c_{2,1} &= m_3 + m_5\\
   m_4 &= (a_{1,2} - a_{2,2}) (b_{2,1} + b_{2,2}) & c_{2,2} &= m_1 + m_2 - m_3 + m_7.\\
   m_5 &= (a_{2,2}) (b_{2,1} - b_{1,1}) \\ 
   m_6 &= (a_{1,1} + a_{1,2})(b_{2,2})\\
   m_7 &= (a_{2,1} - a_{1,1}) (b_{1,1}+ b_{1,2})
\end{alignat*}

Like any bilinear operation, matrix multiplication can be encoded in the
language of tensors.  The rank of matrix multiplication is then equal to the
tensor rank of the matrix multiplication tensor.

  \begin{definition}
   For $n\in \N$ we define the matrix multiplication tensor
   \begin{equation*} 
     \M_{n} = \sum_{i,j,k=1}^{n}E_{i,j}\otimes E_{j,k} \otimes E_{k,i} \in  (K^{n\times n})^{\otimes 3},
   \end{equation*}
   where $E_{i,j}$ stands for a matrix with a $1$ at position $(i,j)$ and zeros
   elsewhere.

   A tensor $T\in (K^{n\times n})^{\otimes 3}$ has rank one if $T=A\otimes B
   \otimes C$ for some $A,B,C \in K^{n\times n}$. The rank of a tensor is the
   smallest number $r$ such that $T$ can be written as the sum of $r$ rank-one
   tensors.

   An $n\times n$ matrix multiplication scheme of rank $r$ is a set of $r$
   rank-one tensors that sum to $\M_{n}$.
 \end{definition}

 Instead of $E_{i,j}\otimes E_{j',k'} \otimes E_{k'',i''}$ we from now on write
 $a_{i,j}\otimes b_{j',k'} \otimes c_{k'',i''}$ for easier readability and to
 highlight the connection to matrix multiplication.

 Strassen's algorithm corresponds to the following decomposition of the $2\times
 2$ matrix multiplication tensor into $7$ rank-one tensors:
 \begin{equation*}\label{eq:strassen}
   \begin{alignedat}{2}
     \M_{2} &= (a_{1,1} + a_{2,2}) \otimes (b_{1,1} + b_{2,2}) \otimes (c_{1,1}+c_{2,2})\\
     &+(a_{1,1}) \otimes (b_{1,2} - b_{2,2}) \otimes (c_{2,1}+c_{2,2})\\
     &+(a_{2,1} + a_{2,2}) \otimes (b_{1,1}) \otimes (c_{1,2}-c_{2,2})\\
     &+(a_{1,2} - a_{2,2}) \otimes (b_{2,1} + b_{2,2}) \otimes (c_{1,1})\\
     &+(a_{2,2}) \otimes (b_{2,1} - b_{1,1}) \otimes (c_{1,1}+c_{1,2})\\
     &+(a_{1,1} + a_{1,2}) \otimes (b_{2,2}) \otimes (c_{2,1}-c_{1,1})\\
     &+(a_{2,1} - a_{1,1}) \otimes (b_{1,1}+ b_{1,2}) \otimes (c_{2,2}).
   \end{alignedat}
 \end{equation*}

 The first two factors in each summand encode the multiplications $m_{1},\ldots
 m_{7}$ in the algorithm and the third factor how they contribute to each
 $c_{i,j}$.  In the tensor notation the indices in the $c$ variables are
 swapped, this makes the symmetries of the algorithm more apparent.  We can
 observe that a cyclic permutation of the factors leaves the first summand
 invariant. Also the remaining summands can be grouped into two blocks of three
 that both are invariant under cyclic permutation of the factors.  If we view
 the factors as matrices, then these two blocks are related by reversing the
 rows and columns of each matrix.  These symmetries can be expressed by the
 action of a symmetry group.

 The matrix multiplication tensor itself has a much larger symmetry group.  De
 Groote~\cite{dg:ovo1} determined this symmetry group and showed that Strassen's
 algorithm is unique up to the action of this group~\cite{dg:ovo2}.  We consider
 the actions of the two groups $S_{3}$ and $\gl_{n}^{\times 3}$ on the space
 $(K^{n\times n})^{\otimes 3}$.  The group $S_{3}$ acts by permuting the three
 factors and transposing them if the permutation is odd.  An element $(U,V,W)
 \in \gl(n)^{\times 3}$ acts on a rank-one tensor $A\otimes B\otimes C$ by the
 so-called sandwiching action $(U,V,W)\cdot A\otimes B\otimes C = UAV^{-1}
 \otimes VBW^{-1} \otimes WCU^{-1}$.  

 To write these actions as a single group action we need to form a semidirect
 product $S_{3}\rtimes \gl(n)^{\times 3}$. Strictly speaking, the symmetry group
 of matrix multiplication is a subgroup of this group, since the action of
 $\gl(n)^{\times 3}$ is not faithful on $(K^{n\times n})^{\otimes 3}$.  However,
 we still refer to it as the symmetry group of matrix multiplication.  Two
 algorithms are considered equivalent if they are related by the action of the
 symmetry group.

 \begin{definition}
 Let $G$ be a subgroup of the symmetry group of matrix multiplication.  A matrix
 multiplication scheme $S$ is called $G$-invariant if for all $g\in G$ we have
 $g\cdot S=S$.
 \end{definition}

 A $G$-invariant matrix multiplication scheme $S$ can be partitioned into orbits
 under the action of $G$.  We can write $S$ as
 \begin{equation*}
   S=\bigcup\{G \cdot t | t\in T\},
 \end{equation*}
 where $T$ is some set of representatives for the orbits of $G$ in $S$.
 
 In the present paper we consider two subgroups of the symmetry group.  All
 decompositions will be invariant under the action of $C_{3} \leq S_{3}$ that
 cyclically permutes the factors.  We also use a $\Z_{2}$ action that reverses
 the rows and columns of every factor matrix (this is conjugation with a
 backwards identity in every factor, and thus lies in $\gl(n)^{\times 3}$).
 Since we apply the same $\gl(n)$ action in each factor, these two actions
 commute and therefore can be combined into an action of the direct product
 $C_{3}\times \Z_{2}$.

 This group action describes the symmetries we observed in Strassen's algorithm
 above.  Thus, Strassen's algorithm is $C_{3}\times \Z_{2}$-invariant and the
 corresponding decomposition can be written as
 \begin{equation*}\label{eq:strassen-sym}
   \begin{alignedat}{2}
     \M_{2} &= (a_{1,1}+a_{2,2})\otimes(b_{1,1}+b_{2,2})\otimes(c_{1,1}+c_{2,2})\\
            &+ C_{3}\times\Z_{2}\cdot (a_{1,1}) \otimes (b_{1,2}-b_{2,2})
              \otimes (c_{2,1}+c_{2,2}).
   \end{alignedat}
 \end{equation*}
 Here the second summand stands for the sum of all orbit elements.

\section{Flip Graph}
\label{sec:flip-graph}
The flip graph algorithm introduced by Kauers and Moosbauer~\cite{km:fgfm} is an
effective method to find upper bounds on the rank of matrix multiplication
tensors~\cite{km:fgfm,km:snnm,aih:afga}.  The main concept is to follow random
walks on a graph whose vertices are matrix multiplication schemes and whose edges
are transformations between them.  Three kinds of transformations have been
introduced: flips, reductions and plus-transitions.  Reductions decrease the
rank of a matrix multiplication scheme, therefore the goal is to find as many
reduction edges as possible.  Flips transform a scheme into another scheme of
the same rank, they are used to explore the graph for reductions.
Plus-transitions increase the rank of a matrix multiplication scheme, they were
introduced by Arai, Ichikawa and Hukushima~\cite{aih:afga} to escape from
components of the flip graph that contain no reductions.

Every transformation replaces two rank-one tensors by one, two, or three
different rank-one tensors, without changing the rest of the tensor.

\begin{definition}
  Let $n\in \N$ and let $S$ be an $n\times n$ matrix multiplication scheme.  Let
  $A\otimes B \otimes C, A'\otimes B'\otimes C'\in S$.

  If $A=A'$ then we call a scheme of the form
  \begin{align*}
    S' &= S \setminus \{A\otimes B\otimes C, A'\otimes B'\otimes C'\}
       \cup \{A\otimes B \otimes (C+C'), A'\otimes (B'-B) \otimes C'\}.
  \end{align*}
  a flip of $S$.
  
  If $A=A'$ and $B=B'$ then we call a scheme of the form
  \begin{align*}
    S' &= S \setminus \{A\otimes B\otimes C, A'\otimes B'\otimes C'\}
       \cup \{A\otimes B \otimes (C+C')\}.
  \end{align*}
  a reduction of $S$.

  We call a scheme of the form
  \begin{align*}
    S' &= S \setminus \{A\otimes B\otimes C, A'\otimes B'\otimes C'\}
       \cup \{(A-A')\otimes B \otimes C, A'\otimes B \otimes (C+C'), A'\otimes (B'-B)\otimes C'\}.
  \end{align*}
  a plus-transition of $S$.

 These definitions apply for arbitrary permutations of $A,B$ and $C$.
\end{definition}

Note that if $S'$ is a flip of $S$ then $S$ is a flip of $S'$ as well.  It is
straightforward to see that flips and reductions of matrix multiplication
schemes are matrix multiplication schemes again.  To show the same for
plus-transitions, we can write them as an inverse reduction followed by a flip:
\begin{align*}
    &A\otimes B\otimes C + A'\otimes B'\otimes C'\\
    &= (A-A')\otimes B\otimes C + A'\otimes B\otimes C + A'\otimes B'\otimes C'\\
    &= (A-A')\otimes B\otimes C + A'\otimes B\otimes (C+C') + A'\otimes (B'-B)\otimes C'.
\end{align*}
  
\emph{Remark.} The original definition of a reduction in~\cite{km:fgfm} is more
general, however every such general reduction can be expressed as a sequence of flips
followed by a reduction defined as here.  In our experiments, we have found that
those general reductions occur extremely rarely; thus, we do not specifically
search for them.

Our new search algorithm is based on the property that we can apply flips,
reductions and plus-transitions to two orbits instead of two single rank-one
tensors, which preserves the symmetry of the scheme.
\begin{theorem}\label{thm:orbit-flip}
  Let $n\in \N$ and let $G\leq S_{3}\rtimes \gl(n)^{\times 3}$ be a finite
  subgroup of the symmetry group of matrix multiplication.  Let $S$ be a
  $G$-invariant $n\times n$ matrix multiplication scheme and let $T$ be a set of
  representatives for the orbits of $G$ in $S$.  Let $A\otimes B\otimes C,
  A\otimes B'\otimes C'$ be elements in $T$ with orbits of size $|G|$. Then
  there are $\lambda,\lambda' \in \Z$ such that
  \begin{align*}
    &S' = S\setminus (G\cdot A\otimes B \otimes C \cup G\cdot A\otimes B' \otimes C')
    \cup G\cdot \lambda A\otimes B \otimes (C+C') \cup G\cdot \lambda' A \otimes (B'-B) \otimes C'
  \end{align*}
  is a $G$-invariant $n\times n$ matrix multiplication scheme.

  We call $S'$ an \emph{orbit flip} of $S$.
\end{theorem}
\begin{proof}
  We first show that for all $g\in G$ we have
  \begin{align*}
    &g\cdot A\otimes B\otimes C + g\cdot A\otimes B'\otimes C'
    = g\cdot A\otimes B\otimes (C+C') + g\cdot A\otimes (B'-B)\otimes C'.
  \end{align*}
  Since we can write $g$ as a product of the form
  \begin{equation*}
    g=(\mathbb{1},U,V,W)(\pi,I_{n},I_{n},I_{n}),
  \end{equation*}
  we consider the two actions separately. Note that the identity
  \begin{align*}
    &A\otimes B \otimes C + A\otimes B' \otimes C'
    = A \otimes B \otimes (C+C') +
      A\otimes (B'-B) \otimes C'
  \end{align*}
  still holds if we apply the same permutation to every rank-one tensor.  Since
  $(X\pm Y)^{T}=X^{T}\pm Y^{T}$, the identity also holds if we transpose every
  factor.

  For the second action we have
  \begin{align*}
    &(U,V,W)\cdot A\otimes B \otimes C + (U,V,W)\cdot A\otimes B' \otimes C'\\
    &=UAV^{-1}\otimes VBW^{-1}\otimes WCU^{-1} + UAV^{-1}\otimes VB'W^{-1}\otimes WC'U^{-1} \\
    &=UAV^{-1}\otimes VBW^{-1}\otimes (WCU^{-1}+WC'U^{-1})\\
    &+ UAV^{-1}\otimes (VB'W^{-1}-VBW^{-1})\otimes WC'U^{-1}\\
    &=UAV^{-1}\otimes VBW^{-1}\otimes W(C+C')U^{-1}\\
    &+ UAV^{-1}\otimes V(B'-B)W^{-1}\otimes WC'U^{-1}\\
    &= (U,V,W)\cdot A\otimes B \otimes (C+C') + (U,V,W)\cdot A\otimes (B'-B) \otimes C'.
  \end{align*}

  Since both $A\otimes B \otimes C$ and $A\otimes B' \otimes C'$ have full
  orbits, we get
  \begin{align*}
    &G\cdot A\otimes B\otimes C + G\cdot A\otimes B'\otimes C'\\
    &= \sum_{g\in G}(g \cdot A\otimes B\otimes C + g\cdot A\otimes B'\otimes C')\\
    &= \sum_{g\in G}(g \cdot A\otimes B\otimes (C+C') + g\cdot A\otimes (B'-B)\otimes C').
  \end{align*}
   
  The orbits of $A\otimes B\otimes (C+C')$ and $A\otimes B'-B)\otimes C'$ need
  not have size $|G|$.  Thus, we set
  \begin{equation*}
  \lambda=\frac{|G|}{|G\cdot A\otimes
    B\otimes (C+C')|} , \lambda'=\frac{|G|}{|G\cdot A\otimes
    (B'-B)\otimes C'|}.
  \end{equation*}
  The orbit-stabilizer theorem guarantees that $\lambda$
  is an integer and that
  \begin{equation*}
  \sum_{g\in G} g\cdot A\otimes B \otimes
  (C+C') = \lambda \sum\Big(G \cdot A\otimes B \otimes (C+C')\Big).
  \end{equation*}
  The same holds for $\lambda'$ and $A\otimes (B'-B) \otimes C$.  A special case
  occurs if $G\cdot A\otimes B \otimes (C+C')=G\cdot A \otimes (B'-B)\otimes C$.
  In this case we double $\lambda$ and set $\lambda'$ to $0$.  Finally, note
  that if $K$ has finite characteristic, then both $\lambda A\otimes B\otimes
  (C+C')$ and $\lambda' A \otimes (B'-B) \otimes C$ can be $0$.
\end{proof}

Since reductions are a special case of flips and plus-transitions are an
inverse reduction followed by a flip, Theorem~\ref{thm:orbit-flip} also holds
for reductions and plus-transitions.  We call the corresponding operations
\emph{orbit reduction} and \emph{orbit plus-transition}.  Note that several
different representatives can allow to flip two orbits in multiple ways.

From Theorem~\ref{thm:orbit-flip} it follows that it is sufficient to apply a
flip, reduction or plus-transition to the representatives of two orbits to
compute the transformation of the whole orbits.  This motivates the following
definition:

\begin{definition}[Flip Graph with Symmetry]
  Let $n\in \N$ and let $G$ be a finite subgroup of the symmetry group of matrix
  multiplication.  Let $V$ be the set of $G$-invariant $n\times n$-matrix
  multiplication schemes.  We define
  \begin{align*}
    &E_{1} = \{ (S,S') \mid S' \text{ is an orbit flip of } S\},\\
    &E_{2} = \{ (S,S') \mid S' \text{ is an orbit reduction of } S\},\\
    &E_{3} = \{ (S,S') \mid S' \text{ is an orbit plus-transition of } S\}.
  \end{align*}

  The graph $(V,E_{1}\cup E_{2} \cup E_{3})$ is called the $(n,n,n)$-flip graph with
  symmetry group $G$.

  We use the notation $(n,n,n)$-flip graph to be consistent with existing
  notation for rectangular matrix multiplication, even though we only consider
  square matrices.
\end{definition}

It has been shown in~\cite{km:fgfm} and in~\cite{aih:afga} that the regular flip
graph is connected if we consider reductions as undirected edges or include
plus-transitions.  This does not hold in the same generality for the flip graph
with symmetry.  If the characteristic of the field is coprime to the group
order, then allowing transformations that replace $A\otimes B\otimes C$ by $k$
copies of $\frac1k A\otimes B\otimes C$ would be sufficient to establish
connectivity of the flip graph.  However, if the group order and the
characteristic of the ground field are not coprime, then certain schemes cannot
be connected in the flip graph.  For example take $K=F_{2}$. Strassen's
algorithm contains the rank-one tensor
\begin{equation*}
  (a_{11}+a_{22})(b_{11}+b_{22})(c_{11}+c_{22}).
\end{equation*}
Any flip, reduction or plus-transition that generates this rank-one tensor would
generate it $6$ times, so it becomes $0$ over $F_{2}$.  We could also allow
orbit flips for smaller orbits, as long as both orbits are generated by the same
group elements, but this does not help here.  Orbits of size two will always
lead to an even coefficient. Orbits of size one or three must be
$C_{3}$-invariant, but a flip can only modify one of the three factors at a
time.  So if we start from any $C_{3}\times \Z_{2}$ invariant $2\times 2$ matrix
multiplication scheme that does not already contain
$(a_{11}+a_{22})(b_{11}+b_{22})(c_{11}+c_{22})$ then there is no path to
Strassen's algorithm in the flip graph.  Thus, we need to give additional
attention to the choice of a suitable starting point.

\section{Search Algorithm}
While the flip graph is too large for exhaustive search, following random walks
allows us to discover matrix multiplication schemes with low rank.  For all our
experiments we used $K=F_{2}$ as a base field.  Although the algorithms we
discover only apply to fields of characteristic $2$ directly, we demonstrate
later how they can be lifted to arbitrary fields.  As symmetry group we use
either $G=C_{3}$ or $G=C_{3}\times \Z_{2}$.

The standard algorithm corresponds to the decomposition of the matrix
multiplication tensor into standard basis elements, i.e. elements of the form
$a_{i,j}\otimes b_{j,k}\otimes c_{k,i}$.  Some of these rank-one tensors are
$C_{3}$-invariant. This is only the case if $i=j=k$.  If $n$ is odd then there
is one $\Z_{2}$ invariant rank-one tensor in the standard algorithm, namely
$a_{(n+1)/2,(n+1)/2}\otimes b_{(n+1)/2,(n+1)/2}\otimes c_{(n+1)/2,(n+1)/2}$, otherwise there
are none.  Such rank-one tensors will not be replaced when we follow a path in
the flip graph.  So, instead of choosing the standard algorithm as a starting
point, we guess some specific rank-one tensors that generate those entries of
the matrix multiplication tensor.  In our experiments we found that considering
$C_{3}$-invariant rank-one tensors whose factors are diagonal matrices works
best.

\begin{definition}\label{def:starting-point}
  Let $n\in \N$ and let $G =C_{3}$ or $G=C_{3}\times \Z_{2}$.  Let $\mathcal{P}$ be a
  partition of $\{1,\ldots,n\}$ such that for every element $P\in \mathcal{P}$, we also
  have $\{n+1-i\mid i\in P\} \in \mathcal{P}$ if $G=C_{3}\times \Z_{2}$.
  Let
  \begin{equation*}
    T=\{(\sum_{i \in P}a_{i,i})\otimes (\sum_{i \in P}b_{i,i}) \otimes (\sum_{i \in P}c_{i,i})\mid P\in \mathcal{P}\}
  \end{equation*}

  Let $S$ be the decomposition of $\M_{n}-\sum_{t\in T} t$ into standard basis
  elements.  Then we call $S\cup T$ the starting point for the diagonal
  partition $\mathcal{P}$.
\end{definition}

A starting point chosen this way can be partitioned into full orbits: 

\begin{lemma}
  Let $G = C_{3}$ or $G=C_{3}\times \Z_{2}$.  Let $S\cup T$ be the starting
  point for a diagonal partition $\mathcal{P}$.  Then $S\cup T$ is a $G$-invariant matrix
  multiplication scheme and the orbit of every element in $S$ has size $|G|$.
\end{lemma}
\begin{proof}
  $T$ is $G$-invariant by construction.  Since $\M_{n}-\sum_{t \in T}t$ is
  $G$-invariant, $S$ must b $G$-invariant as well.

  To show that $S$ consists of full orbits we have to show that for any $s\in S,
  g\in G\setminus \{e\}$ we have $g\cdot s \neq s$, where $e$ is the identity
  element of $G$.  $S$ consists of two types of elements, those that come from
  $\M_{n}$ are of the form $a_{i,j}\otimes b_{j,k} \otimes c_{k,i}$ and those
  that are correction terms for $T$ are of the form $a_{i,i}\otimes b_{j,j}
  \otimes c_{k,k}$.  Note that in both kinds of elements we cannot have $i=j=k$.
  It is sufficient to show the claim for a generator $g$ of $G$.

  In the case $G=C_{3}$, we can choose $g$ such that
\begin{equation*}
  g\cdot a_{i,j}\otimes
  b_{j,k}\otimes c_{k,i} = a_{j,k}\otimes b_{k,i}\otimes c_{i,j}.
\end{equation*}
In this case $g\cdot a_{i,j}\otimes b_{j,k} \otimes c_{k,i} = a_{i,j}\otimes
b_{j,k}\otimes c_{k,i}$ implies $i=j=k$ as does $g\cdot a_{i,i}\otimes b_{j,j}
\otimes c_{k,k} = a_{i,i}\otimes b_{j,j}\otimes c_{k,k}$.

  If $G=C_{3}\times \Z_{2}$, then we choose $g$ such that
  \begin{equation*}
    g\cdot a_{i,j}\otimes
    b_{j,k}\otimes c_{k,i} = a_{n+1-j,n+1-k}\otimes b_{n+1-k,n+1-i}\otimes c_{n+1-i,n+1-j}.
  \end{equation*}
  So $g\cdot a_{i,j}\otimes b_{j,k}\otimes c_{k,i} = a_{i,j}\otimes
  b_{j,k}\otimes c_{k,i}$ implies $i=n+1-j$, $j=n+1-k$ and $k=n+1-i$, so
  $i=j=k$. The same argument works when $s=a_{i,i}\otimes b_{j,j} \otimes
  c_{k,k}$.
\end{proof}

In our experiments we rarely encountered orbits that became smaller during the
search process.  So for a given starting point $S$ with diagonal partition $\mathcal{P}$
we expect to only discover schemes of rank $k\cdot |G| + |\mathcal{P}|$ where $k$ is the
number of full orbits in the result.  

If we look for a $C_{3} \times \Z_{2}$-invariant $2\times 2$ matrix
multiplication algorithm, like Strassen, then we can choose between the two
diagonal partitions $\{\{1,2\}\}$ and $\{\{1\},\{2\}\}$. Only the first allows
to find Strassen's algorithm.  The corresponding starting point is
\begin{align*}
  \{&(a_{1,1}+a_{2,2})\otimes (b_{1,1}+b_{2,2}) \otimes (c_{1,1}+c_{2,2}),\\
  &C_{3}\times \Z_{2} \cdot a_{1,1}\otimes b_{1,1} \otimes c_{2,2},\\
  &C_{3}\times \Z_{2} \cdot a_{1,1}\otimes b_{1,2} \otimes c_{2,1}\}.
\end{align*}
Here the first element comes from the diagonal partition, the second element are
the correction terms needed to remove the error introduced by the first element
and the third element generates the non-diagonal entries of the $2\times 2$
matrix multiplication tensor.

\begin{figure}
  \includegraphics[scale=0.35]{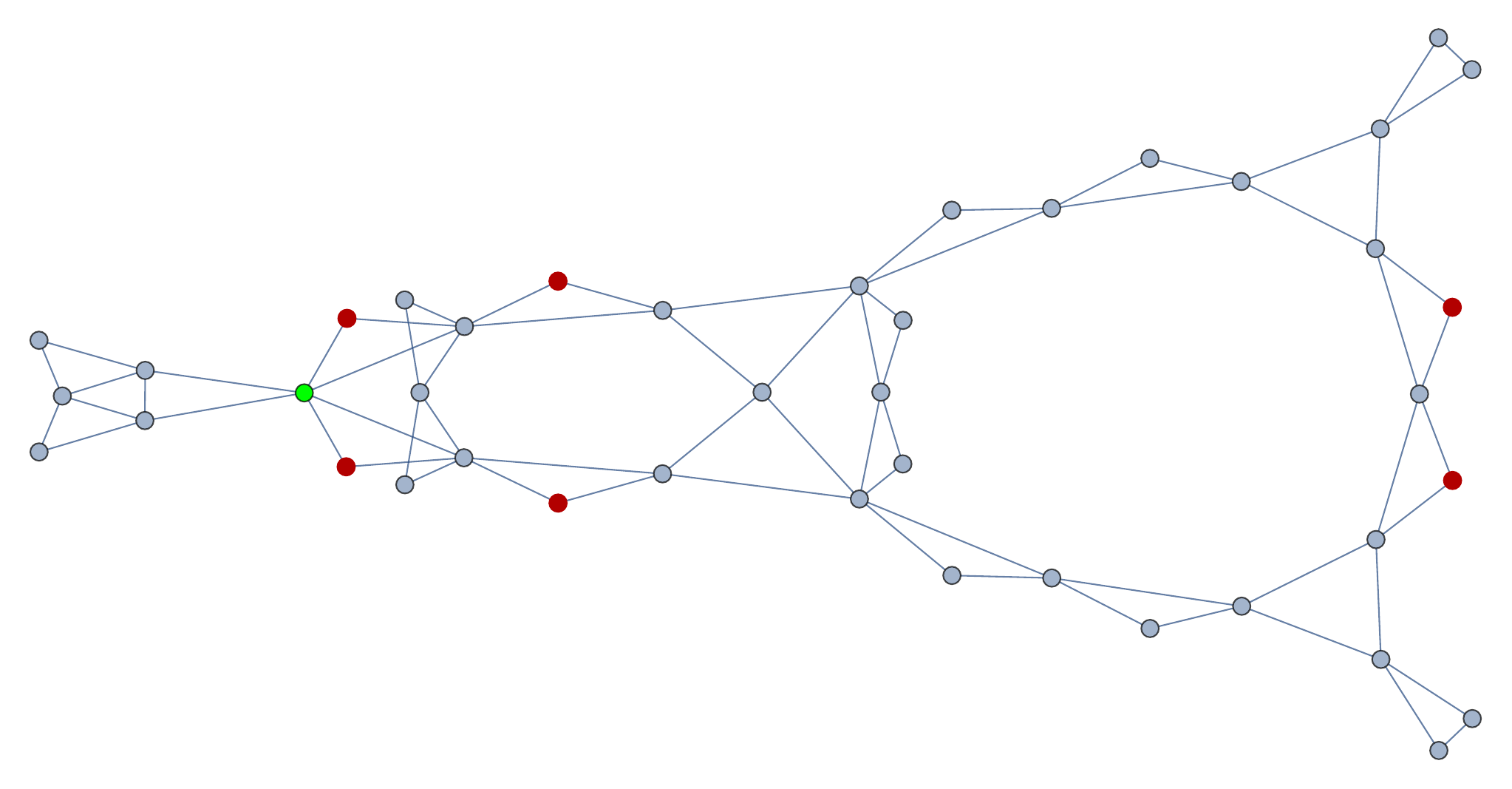}
  \caption{The component of the $C_{3}\times \Z_{2}$-symmetric $(2,2,2)$ flip
    graph that contains the starting point with diagonal partition
    $\{\{1,2\}\}$, marked in green.  Vertices that can be reduced to Strassen's
    algorithm are marked in red. }
  \label{fig:222flip-graph}
\end{figure}

The $2 \times 2$ case is small enough that we can compute all schemes that are
reachable by flips from this starting point.  The corresponding component of the
graph is shown in Figure~\ref{fig:222flip-graph}.  We can observe that the graph
has a reflection symmetry.  This symmetry corresponds to the $S_{3}$ action that
transposes all factors and swaps two of them.  So each scheme form the top half
is equivalent to the corresponding scheme from the bottom part.

The figure shows that the starting point has a neighbor that can be reduced.
Applying the suitable flip gives
\begin{align*}
  \{&(a_{1,1}+a_{2,2})\otimes (b_{1,1}+b_{2,2}) \otimes (c_{1,1}+c_{2,2}),\\
  &C_{3}\times \Z_{2} \cdot a_{1,1}\otimes (b_{1,1}+b_{1,2}) \otimes c_{2,2},\\
  &C_{3}\times \Z_{2} \cdot a_{1,1}\otimes b_{1,2} \otimes (c_{2,1}+c_{2,2})\}.
\end{align*}
It is not immediately clear here that a reduction is possible.  The reduction
becomes apparent if we choose a different representative for the third orbit:
\begin{align*}
  \{&(a_{1,1}+a_{2,2})\otimes (b_{1,1}+b_{2,2}) \otimes (c_{1,1}+c_{2,2}),\\
  &C_{3}\times \Z_{2} \cdot a_{1,1}\otimes (b_{1,1}+b_{1,2}) \otimes c_{2,2},\\
  &C_{3}\times \Z_{2} \cdot a_{2,1}\otimes (b_{1,1}+b_{1,2}) \otimes c_{2,2}\}.
\end{align*}
Since the second and third element share two factors they can be reduced to a
single orbit, yielding Strassen's algorithm over $F_{2}$:
\begin{align*}
  \{&(a_{1,1}+a_{2,2})\otimes (b_{1,1}+b_{2,2}) \otimes (c_{1,1}+c_{2,2}),\\
  &C_{3}\times \Z_{2} \cdot (a_{1,1}+a_{2,1})\otimes (b_{1,1}+b_{1,2}) \otimes c_{2,2}\}.
\end{align*}

For larger cases we use a random search procedure on the flip graph.  To
maximize efficiency in our implementation we choose not to look for reductions
explicitly.  Instead we rely on flips that realize a reduction by transforming
an element to a $0$.  The reason that we do not miss too many reductions this
way is that most neighbors of a reducible scheme will still allow the same
reduction step.

Our search strategy is as follows.  We follow random flips until a reduction is
performed or a certain number of flips $M$ is reached. If we did not encounter a
reduction then we perform a random plus transition and continue the process.  We
keep track of the minimal rank that is found during one run.  If a given limit
$L$ on the number of flips is reached without encountering a new minimal rank we
terminate the search.  Whenever a new minimal rank is found, we reset the limit
to ensure that the search is not terminated because too much time was spent on
higher ranks.  

The exact search routine is described in the following pseudo
code:

 \begin{algorithm}\label{alg:path}
   Input: A symmetry group $G$, a $G$-invariant matrix multiplication scheme
   $S$, a limit $L$ for the path length, a target rank $R$ and a plus-transition parameter $M$.\\
   Output: A $G$-invariant matrix multiplication scheme of rank $R$ or $\bot$.
   \step10 $\l=0, m=M, r=\rank(S)$
   \step20 while $\l<L$ do:
   \step31 If $S$ has no orbit flips then:
   \step42 Return $\bot$.
   \step51 Else:
   \step62 Set $S$ to a random orbit flip of~$S$.
   \step71 If $0 \in S$ then:
   \step82 Remove $0$ from $S$.
   \step{9}2 If $\rank(S) < r$ then:
   \step{10}3 $r=rank(S)$
   \step{11}3 If $r\leq R$ then return $S$.
   \step{12}3 $\l=0$
   \step{13}2 $m=\l+M$
   \step{14}1 if $\l\geq m$ then:
   \step{15}2 Set $S$ to a random orbit plus-transition of~$S$.
   \step{16}2 $m=\l+M$
   \step{17}1 $\l=\l+1$
   \step{18}0 Return $\bot$.
 \end{algorithm}

 \section{Results}

 For $n=5$ we used $G=C_{3}$ and the starting point given by the diagonal
 partition $\{\{1,5\},\{2,4\},\{3\}\}$.  With a flip limit of $10^{8}$ and a
 plus-transition parameter of $5\cdot 10^{4}$, we can find a scheme of rank $93$ within
 few minutes on a laptop.

 For $n=6$ we used $G=C_{3}\times \Z_{2}$ and the starting point for the diagonal
 partition $\{\{1,2\},\{3,4\},\{5,6\}\}$.  With a flip limit of $10^{9}$ and a
 plus-transition parameter of $5\cdot 10^{5}$, we can find a scheme of rank $153$.  The
 search took around half an hour on a machine with $160$ CPU cores, amounting to
 about $80$ hours of CPU time.

 For $n=3$ we can match the upper bound $23$ with $G=C_{3}$ and any diagonal
 partition into two parts.  For $n=4$ we find a $C_{3}\times \Z_{2}$-invariant
 scheme with rank $49$.  However, we are not able to find any symmetric scheme
 with a rank less than $49$, while schemes without symmetry exist with rank
 $47$~\cite{fbh+:dfmm}.

 All schemes we found could be lifted to coefficients in $\Z$ as described in
 Section~\ref{sec:lifting}.  The schemes and the code for the search can
 be downloaded from
 \begin{center}
   \url{https://github.com/jakobmoosbauer/symmetric-flips}.
 \end{center} 
 
\section{Lifting}\label{sec:lifting}
To convert a matrix multiplication scheme over $F_{2}$ to a scheme that works
over arbitrary fields, we apply Hensel lifting~\cite{vg:mca} as also done
in~\cite{km:fgfm}.  We use a modification of the lifting procedure that
preserves the cyclic symmetry of the schemes.  To this end we consider a cyclic
version of the Brent equations~\cite{br:afmm}.  Let
\begin{equation*}
S=\{((\alpha_{i,j}^{(\l)}))\otimes
((\beta_{j,k}^{\l}))\otimes((\gamma_{k,i}^{\l}))\mid \l=1,\ldots ,k\}
\end{equation*}
be a set of rank-one tensors and let $T$ be the rank-one tensors defined by the
diagonal partition.  Remember that the elements of $T$ are $C_{3}$-invariant.
We have that $\bigcup \{C_{3}\cdot s\colon s \in S\}\cup T$ is a $C_{3}$-invariant matrix
multiplication scheme if and only if the equations
\begin{align*}
  &\sum_{\l=1}^{k}\alpha_{i_{1},i_{2}}^{(\l)}\beta_{j_{1},j_{2}}^{(\l)}\gamma_{k_{1},k_{2}}^{(\l)} +
    \alpha_{j_{1},j_{2}}^{(\l)}\beta_{k_{1},k_{2}}^{(\l)}\gamma_{i_{1},i_{2}}^{(\l)} +
    \alpha_{k_{1},k_{2}}^{(\l)}\beta_{i_{1},i_{2}}^{(\l)}\gamma_{j_{1},j_{2}}^{(\l)} 
  = \Big(\M_{n}-\sum_{t\in T}t\Big)_{i_{1},i_{2},j_{1},j_{2},k_{1},k_{2}}
\end{align*}
for $i_{1},i_{2},j_{1},j_{2},k_{1},k_{2}=1,\ldots,n$ are satisfied.  In our
experiments it was not necessary to lift the coefficients in the elements of
$T$.  Therefore, they are considered as constants rather than variables.

A solution to this system modulo $2^{m}$ for some $m\in \N$ can be seen as an
approximation of order $m$ to a solution lying in the $2$-adic integers.  Given
such a solution we can make an ansatz with undetermined coefficients for the
approximation to order $m+1$, plug this ansatz into the equations and reduce mod
$2^{m+1}$.  The result will always be divisible by $2^{m}$, this gives a linear
system over $F_{2}$ for the undetermined coefficients in the ansatz, which can
be solved with linear algebra.  If the system has no solution, then this proves
that there is no refinement of the approximation to order $m+1$.  If the system
has solutions we pick one and proceed to the next order.

There is no guarantee for the approximation to converge to a solution, but in
our experiments all candidates converged to an integer solution within a few
steps.  This is noteworthy, since several of the matrix multiplication schemes
discovered over $F_{2}$ in prior work~\cite{fbh+:dfmm,km:fgfm,aih:afga} can
provably not be lifted.  Moreover, when applying Hensel lifting to schemes
discovered in~\cite{km:fgfm}, only some of the candidates can be lifted to a
solution over $\mathbb{Q}$ and even fewer admit lifting to a solution over $\Z$.  These
findings suggest that liftability is a property that is much more common among
symmetric matrix multiplication schemes.

\section{Identifying Starting Points}\label{sec:starting-points}
Finding a good starting point is mostly a process of trial and error.  Once a
scheme of rank $r$ is found, the new task becomes to find a scheme of rank less
than $r$.  If we keep the same starting point we can only search for schemes of
rank $r-|G|$.  When this search is unsuccessful, the question is whether we are
able to find schemes of rank between $r$ and $r-|G|$.  For every rank we want to
test we also need to test several possible starting points.

However, in the case of $6\times 6$ matrices we were able to use the starting
point for the diagonal partition into a single part to identify a good starting
point for a $3$-part diagonal partition.  We first discovered schemes of rank
$157$ by using the starting point where
\begin{equation}\label{eq:start6x6-157}
\begin{aligned}
  &T= (a_{1,1}+a_{2,2}+a_{3,3}+a_{4,4}+a_{5,5}+a_{6,6}) \\
  &\otimes (b_{1,1}+b_{2,2}+b_{3,3}+b_{4,4}+b_{5,5}+b_{6,6}) \\
  &\otimes (c_{1,1}+c_{2,2}+c_{3,3}+c_{4,4}+c_{5,5}+c_{6,6})\}.
\end{aligned}
\end{equation}
Many of the resulting rank-$157$ schemes contained the following orbit, which
stands out because all factors are diagonal matrices:
\begin{equation}\label{eq:diag6x6}
  G\cdot(a_{1,1}+a_{2,2}) \otimes (b_{1,1}+b_{2,2}+b_{3,3}+b_{4,4}) \otimes (c_{3,3}+c_{4,4}+c_{5,5}+c_{6,6}).
\end{equation}
After expanding this term one can observe that~(\ref{eq:start6x6-157})
and~(\ref{eq:diag6x6}) can be replaced by the three rank-one tensors
\begin{align*}
  &(a_{1,1}+a_{2,2})\otimes (b_{1,1}+b_{2,2})\otimes (c_{1,1}+c_{2,2}),\\
  &(a_{3,3}+a_{4,4})\otimes (b_{3,3}+b_{4,4})\otimes (c_{3,3}+c_{4,4}),\\
  &(a_{5,5}+a_{6,6})\otimes (b_{5,5}+b_{6,6})\otimes (c_{5,5}+c_{6,6})
\end{align*}
yielding a scheme of rank~$153$.  Once we had identified a good starting point,
we were also able to find schemes of rank~$153$ from scratch.

\section{Implementation}

The algorithm was implemented in Python, with an inner solver routine which is
implemented in optimized C++.  Given that our goal is to search schemes for
multiplying $5\times 5$ matrices and larger, the search space is vast, even
taking the symmetry into consideration. Given also that both reductions and
plus-transitions are very rare occurrences, the primary requirement for the
software implementation was that the basic flip operation should be fast as
possible.  To facilitate that requirement, we do not search for nearby
reductions directly.  Instead we rely on random flips that produce a zero
element, which essentially are reductions.  While this might lead to not
following some potential reductions, it allows us to explore much longer paths
in the flip graph.

To the solver we only pass the rank-one tensors that have full orbits.  For
every rank-one tensor $A\otimes B \otimes C$ only the matrix $A$ needs to be
stored and manipulated, the other factors can be reconstructed with the cyclic
symmetry.  This allows us to perform three flips at the cost of one.  Our
implementation supports matrices of size up to $8\times 8$, which can be stored
effectively in a single 64-bit integer.

We use the following data structures to keep track of the available flips at any
time:
\begin{itemize}
\item A list that stores every unique factor matrix $A$ and in which rank-one
  tensors they appear.  This needs to support constant time insertion and
  deletion and is implemented using a hash map. 
\item A list of those unique factor matrices that occur twice or more.  This
  needs to support constant time insertion, deletion and random sampling and is
  implemented using a hash map and list combination.
\end{itemize}

We found that for this purpose a custom hashed data structure significantly
outperforms those provided by the C++ standard library.

For each flip operation, we need to choose a flip at random from those
available, manipulate the corresponding elements, depending on the symmetry
group and then update the list of available flips.  These operations are very
efficient and speeds of between 1 and 6 billion flips per minute were achieved,
depending on the input and the hardware used.

\section{Conclusion}
We introduced a flip graph algorithm with symmetry constraints, significantly
reducing the size of the search space compared to former flip graph approaches.
By forcing invariance under the groups $C_{3}$ or $C_{3}\times \Z_{2}$ our
method not only facilitates the search process but also enhances the ability to
lift schemes from $F_{2}$ to $\Z$.  This approach led to the discovery of new
schemes for $5\times 5$ and $6\times 6$ matrices. We improve the previously best
known bounds $97$~\cite{km:fgfm} and $160$~\cite{sm:tbca} to $93$ and $153$.  A
further improvement on the $5\times 5$ case by two multiplication steps or the
$6\times 6$ case by one multiplication step would provide an algorithm that is
asymptotically faster than Strassen's algorithm, as $\log_{5}91=2.803$ and
$\log_{6}152=2.804$.

  \par\medskip\noindent\textbf{Acknowledgment.}
  The authors acknowledge the use of the Batch Compute System in the Department
  of Computer Science at the University of Warwick, and associated support
  services, in the completion of this work.  We thank Christian Ikenmeyer and
  Manuel Kauers for helpful discussions on the topic of this paper.
  Additionally, we thank Manuel Kauers for providing access to high-memory
  computing facilities at the Institute for Algebra at the Johannes Kepler
  University.
 
 \bibliographystyle{plain}
 \bibliography{bibliography}
 
\end{document}